\documentclass[10pt, a4paper, twocolumn]{IEEEtran}

\usepackage{amsmath, amssymb, comment,cite}
\usepackage{graphicx, color, tabu}
\usepackage[latin1]{inputenc}
\usepackage{relsize, bbm, mathtools}
\usepackage{algorithm,algorithmic}
\usepackage{balance}
\usepackage{epstopdf}
\usepackage[margin=14.0mm,top=14.5mm,bottom=14.5mm]{geometry}

\usepackage[nodisplayskipstretch]{setspace}
\newtheorem{lemma}{Lemma}
\newtheorem{remark}{Remark}
\newtheorem{theorem}{Theorem}
\newtheorem{corollary}{Corollary}

\newtheorem{proof}{Proof}

\newcounter{mytempeqcounter}

\newcommand{\bigformulatop}[2]{%
	\begin{figure*}[!t]
		\normalsize
		\setcounter{mytempeqcounter}{\value{equation}}
		\setcounter{equation}{#1}
		#2
		
		\setcounter{equation}{\value{mytempeqcounter}}
		\hrulefill
		\vspace*{4pt}
	\end{figure*}
}

\newcommand{\qs}{{\bf s}}

\newcommand{\qv}{{\bf v}}
\newcommand{\qw}{{\bf w}}
\newcommand{\qx}{{\bf x}}
\newcommand{\qy}{{\bf y}}
\newcommand{\qz}{{\bf z}}

\newcommand{\qA}{{\bf A}}

\newcommand{\qF}{{\bf F}}

\newcommand{\qH}{{\bf H}}
\newcommand{\qI}{{\bf I}}

\newcommand{\qT}{{\bf T}}

\newcommand{\Nguard}{{N_{\mathsf{guard}}}}

\newcommand{\Sn}{\sigma_n^2}
\newcommand{\diag}{{\mathsf{diag}}}
\newcommand{\Trace}{{\mathsf{tr}}}

\newcommand{\hmki}{h_{pq,i}}
\newcommand{\hmkihat}{\hat{h}_{pq,i}}

\newcommand{\hmkjhatc}{\hat{h}_{pq,j}^*}

\newcommand{\tmki}{\tau_{pq,i}}
\newcommand{\nmki}{\nu_{pq,i}}
\newcommand{\Lmk}{L_{pq}}
\newcommand{\Lmkp}{L_{pq'}}
\newcommand{\Bmki}{\beta_{pq,i}}
\newcommand{\Lmki}{\ell_{pq,i}}

\newcommand{\Kmki}{k_{pq,i}}

\newcommand{\kpmki}{\kappa_{pq,i}}

\newcommand{\Kpq}{k_{pq}}
\newcommand{\Lpq}{\ell_{pq}}

\IEEEoverridecommandlockouts
\allowdisplaybreaks
\graphicspath{{figures/}}

\begin{document}

\title{\fontsize{0.81cm}{1cm}\selectfont  When Cell-Free Massive MIMO Meets OTFS Modulation: The Downlink Case }

\author{Mohammadali Mohammadi$^\dag$, Hien Quoc Ngo$^\dag$, and  Michail Matthaiou$^\dag$\\
\small{
$^\dag$Centre for Wireless Innovation (CWI), Queen's University Belfast, U.K.\\
Email:\{m.mohammadi, hien.ngo, m.matthaiou\}@qub.ac.uk
}}\normalsize

\maketitle
\begin{abstract}
We provide a performance evaluation of orthogonal time frequency space (OTFS) modulation in cell-free massive MIMO (multiple-input
multiple-output) systems.  By leveraging the inherent sparsity of the delay-Doppler (DD) representation of time-varying channels, we apply the embedded pilot-aided channel estimation method with reduced guard intervals and derive the minimum mean-square error estimate of the channel gains from received uplink pilots at the access points (APs). Each AP applies conjugate beamforming to transmit data to the users. We derive a closed-form expression for the individual user downlink throughput as a function of the numbers of APs, users and DD channel estimate parameters. We compare the OTFS performance with that of orthogonal frequency division multiplexing (OFDM) at high-mobility conditions. Our findings reveal that with uncorrelated shadowing, cell-free massive MIMO with OTFS modulation achieves up to $35\%$ gain in $95\%$-likely per-user throughput, compared with the OFDM counterpart. Finally, the increase in the per user throughput is more pronounced at the median rates over the correlated shadowing scenarios.
\let\thefootnote\relax\footnotetext{This work was supported by a research grant from the Department for the Economy Northern Ireland under the US-Ireland R\&D Partnership Programme.}
\end{abstract}


\vspace{0em}
\section{Introduction}~\label{Sec:Intro}
Future beyond-5G (B5G) wireless communication networks will provide ultra-reliable services as well as ubiquitous connectivity  for a range of emerging mobile applications, including vehicle-to-vehicle (V2V) communications, high-speed railways, and unmanned aerial vehicles. However, wireless channels in high mobility environments are inherently linear time-variant fading channels, which are so called \emph{doubly-selective} channels~\cite{Jakes}. In practice, communications over these channels suffer from severe Doppler spread. Therefore, the application of the widely adopted OFDM modulation in 4G and 5G is no longer a viable option for high-mobility scenarios, as the orthogonality between the sub-carriers breaks down due to inter-carrier interference.

Recently, Hadani \textit{et al} developed a new two-dimensional (2D) modulation, referred to as orthogonal time-frequency space (OTFS) modulation, which has shown significant efficiency in tackling the high Doppler problems occurring in OFDM modulation~\cite{Hadani:WCNC:2017}. By invoking the 2D inverse symplectic finite Fourier transform (ISFFT),  OTFS multiplexes the information symbols in the delay-Doppler (DD) domain rather than in the time-frequency (TF) domain as in OFDM modulation. More specifically, through the 2D transformation from the DD domain to the TF domain, each information symbol will span the entire TF domain channel over an OTFS frame. Therefore,  OTFS efficiently exploits the potential of full-diversity, which is the key for supporting ultra-reliable communications~\cite{Wei:WC:2021}. More importantly, with the domain transformation performed in OTFS modulation, rapidly time-varying channels in the TF domain are converted into quasi-stationary channels in the DD domain,  which exhibits a sparse and stable property. This sparse model facilitates significantly the channel estimation and data detection process for wireless receivers in high-mobility environments~\cite{Raviteja:TWC:2018,Raviteja:TVT:2019}. Therefore, there has been an upsurge of interest in applying OTFS to different wireless communication systems, including massive multiple-input multiple-output (MIMO) systems~\cite{Wang:JSAC:2020,Dobre:JSAC:2021,Feng:ICC:2021,Shi:TWC:2021}.

Cell-free massive  MIMO has been recently recognized as an alternative to co-located massive MIMO  for future wireless networks owing to its substantial improvement of connectivity, spectral and energy efficiencies~\cite{Hien:cellfree}. In cell-free massive MIMO, there are no cell boundaries and a large number of access points (APs) are distributed over a large geographic area and jointly serve many user equipments with different speed profiles. While a large body of research has delved into cell-free massive MIMO, they mostly consider flat-fading channels.  A few recent works in the literature focus on the performance of cell-free massive MIMO over frequency-selective fading channels~\cite{Nguyen:cellfree,Schotten:cellfree}. In~\cite{Nguyen:cellfree}, the authors analyzed the uplink achievable spectral efficiency of a frequency-selective cell-free massive MIMO system under the Wiener phase noise process and with single-carrier transmission. The work in~\cite{Nguyen:cellfree} has been extended to multi-carrier transmission in~\cite{Schotten:cellfree}, where a user-specific resource allocation method was proposed. However, the work of~\cite{Nguyen:cellfree,Schotten:cellfree} cannot accommodate high Doppler spread applications. In order to serve  high-mobility users with time-variant channels, the integration of cell-free massive MIMO and OTFS modulation is expected to further improve the network performance. To the authors' best knowledge, the consolidation of OTFS modulation with cell-free massive MIMO has not been reported before. Thus, this paper will focus on the downlink achievable rate analysis of OTFS in cell-free massive MIMO. The main contributions of our work are as follows:
\begin{itemize}
\item We apply the embedded pilot-aided channel estimation with reduced guard interval and derive the  minimum mean-square error (MMSE) estimate of the channel gains at the APs.
\item We investigate the achievable downlink rate of OTFS modulation in cell-free massive MIMO with conjugate beamforming. We derive closed-form expression for the individual user downlink throughput for finite numbers of APs and users, taking into account the effects of channel estimation errors.
\item Our findings demonstrate that a significant performance improvement can be achieved  by the OTFS  over the OFDM modulation in cell-free massive MIMO over high-mobility channels.
\end{itemize}

\textit{Notation:} We use bold upper case letters to denote matrices, and bold lower case letters to denote vectors; the superscripts $(\cdot)^*$ and $(\cdot)^\dag$ stand for the conjugate and conjugate-transpose, respectively; $\Trace(\qA)$ denotes the trace of $\qA$; $\diag(a_1, a_2, \dots , a_n)$ denotes a square $n\times n$ diagonal matrix whose element in the $i$th row and $i$th column is $a_i, i = 1, \ldots , n$; the matrix $\qF_X=\big(\frac{1}{\sqrt{X}} e^{-j2\pi\frac{k\ell}{X}}\big)_{k,\ell=0,\cdots,X-1}$
denotes the unitary DFT matrix of dimension $X \times X$; the operator $\otimes$ denotes the Kronecker product of two matrices; $(\cdot)_N$ denotes the modulo $N$ operation; $\mathfrak{R}[\cdot]$ denotes the real part of the quantity within the brackets; $\mathbb{N}[a,b]$ represents the set of natural numbers ranging from $a$ to $b$; finally, $\mathbb{E}\{\cdot\}$ denotes the statistical expectation.

\vspace{-0.0em}
\section{System Model}~\label{Sec:SysModel}
We consider a cell-free massive MIMO system consisting of $M_a$ APs and $K_u$ users.  The APs and users are all equipped with single antenna, and they are randomly located in a large area. APs are connected to a central processing unit (CPU) via a backhaul network. The users are assumed to move at different speeds, thus the channels between the AP and users experience doubly-selective fading. An OTFS frame is divided into two phases: uplink payload transmission with channel estimation, and downlink payload transmission. Uplink and downlink transmissions are underpinned by time-division duplex (TDD) operation.

\emph{OTFS Modulation and Channel Model:} Consider an OTFS system with $M$ sub-carriers having $\Delta f$ bandwidth each, and $N$ symbols having $T$ symbol duration. Therefore, the total bandwidth of the system is $M\Delta f$ and $NT$ is the duration of an OTFS block. Modulated data symbols of the $q$th user $\{x_q[k,\ell] , k\in \mathbb{N}[0,N-1], \ell\in \mathbb{N}[0,M-1]\}$ are arranged over DD lattice $\Lambda=\left\{\frac{k}{NT}, \frac{\ell}{M\Delta f}\right\}$, where $k$ and $\ell$ represent the Doppler shift and delay index, respectively. Data symbols $x_q[k,\ell]$ are firstly converted to $X_q[n,m]$ in the TF domain through an ISFFT according to
\vspace{-0.2em}
\begin{align}~\label{eq:Xtf}
X_q[n,m] = \frac{1}{\sqrt{MN}}\sum_{k=0}^{N-1}\sum_{\ell=0}^{M-1} x_q[k,\ell] e^{j2\pi(\frac{nk}{N}-\frac{m\ell}{M})},
\end{align}
where $n\in\mathbb{N}[0,N-1]$, and $m\in\mathbb{N}[0,M-1]$. Accordingly, by using Heisenberg transform, $X_q[n,m]$ is converted to a time domain signal as
\vspace{-0.2em}
\begin{align}~\label{eq:st}
s_q(t)=\sum_{n=0}^{N-1}\sum_{m=0}^{M-1}
X_q[n,m] g_{tx}(t-nT)
e^{j2\pi m \Delta f(t-nT)},
\end{align}
where $g_{tx}(t)$ is the transmitter pulse of duration $T$.

The DD channel between the $q$th user and the $p$th AP is given by~\cite{Raviteja:TWC:2018}
\vspace{-1.3em}
\begin{align}~\label{eq:hmk}
h_{pq}(\tau,\nu) = \sum_{i=1}^{\Lmk}\hmki\delta(\tau-\tmki)\delta(\nu-\nmki),
\end{align}
where $\Lmk$ denotes the number of paths from the $q$th user to the $p$th AP, $\tmki$, $\nmki$, and  $\hmki$ denote the delay, Doppler shift, and the channel gain, respectively, of the $i$th path of the $q$th user to the $p$th AP. The complex channel gains $\hmki$  for different $(pq, i)$ are independent and identically distributed (i.i.d.) RVs with $\hmki\sim\mathcal{CN}(0,\Bmki)$. The delay and Doppler shift for the $i$th path are given by $\tmki = \frac{\Lmki}{M\Delta f}$ and $\nmki = \frac{\Kmki+\kpmki}{N T}$, respectively, where $\Lmki \in \mathbb{N}[0, M-1]$ and $\Kmki\in \mathbb{N}[0, N-1]$ are the delay index and Doppler index of the $i$th path, and $\kpmki\in (-0.5,0.5)$ is a fractional Doppler associated with the $i$th path.  Let $\Kpq$ and $\Lpq$ denote the delay and Doppler taps corresponding to the largest delay and Doppler between the $q$th user and the $p$th AP. We note that the typical value of the sampling time $1/(M\Delta f)$ is usually sufficiently small in the delay domain. Hence, the impact of fractional
delays in typical wideband systems can be neglected~\cite{Raviteja:TWC:2018}.

\vspace{-0.4em}
\subsection{Uplink Payload Data Transmission}~\label{Sec:ULdata}
In the uplink, all $K_u$ users simultaneously send their data to the APs. The received signal at the $p$th AP is expressed as
\vspace{-0.2em}
\begin{align}~\label{eq:rt}
r_p(t)\!=\!
\sum_{q=1}^{K_u}
\!\int\!\!\!\!\int \!\!
h_{pq}(\tau,\nu)s_q(t\!-\!\tau)e^{j2\pi\nu(t\!-\!\tau)} d\tau d\nu\! +\! w_p(t),
\end{align}
where $w_p(t)$ represents the noise signal in the time domain following a stationary Gaussian random process and we have $w_p(t)\sim\mathcal{CN}(0,\Sn)$ with $\Sn$ denoting the noise variance. The received signal is processed via a Wigner transform, implementing a receiver filter with an impulse response $g_{rx}(t)$ followed by a sampler, to obtain the received samples $\{Y_p[n,m], n\in\mathbb{N}[0,N-1], m\in\mathbb{N}[0,M-1]$\} in the TF domain
\begin{align}~\label{eq:Ymn}
Y_p[n,m]=\int r_p(t) g_{rx}(t-nT)
e^{-j2\pi m \Delta f (t-nT)} dt.
\end{align}

Finally by applying a SFFT to $Y_p[n,m]$, assuming that  practical non-ideal rectangular waveforms are used for transmit and receive pulse shaping filters~\cite{Raviteja:TWC:2018}, the received signal at the $p$th AP in the DD domain can be written as
\vspace{-0.2em}
\begin{align}~\label{eq:yAP:rec:frac}
y_p[k,\ell]
\!&=\!\sqrt{\rho_u}
\sum_{q=1}^{K_u}\!
\sum_{k'=0}^{\Kpq}
\sum_{\ell'=0}^{\Lpq}
b[k',\ell']
\!\!
\sum_{c=-N/2}^{N/2-1}\!
\!
{h}_{pq}[k',\ell']\alpha[k,l,c]\nonumber\\
&\hspace{2em}
\times x_q[(k\!-\!k'\!+\!c)_N,(\ell\!-\!\ell')_M]
\!+\! w_p[k,\ell],
\end{align}
where $\rho_u$ is the normalized uplink signal-to-noise ratio (SNR); $b[k',\ell']\in\{0,1\}$ is a path indicator, i.e., $b[k',\ell']=1$ indicates there is a path with Doppler tap $k'$ and delay tap $\ell'$, otherwise $b[k',\ell']=0$ (i.e., $\sum_{k'=0}^{\Kpq}\sum_{\ell'=0}^{\Lpq} b[k',\ell']=\Lmk$). Moreover $\alpha[k,l,c]$ is given by~\cite{Raviteja:TWC:2018}
\vspace{-0.2em}
\begin{align}
\alpha[k,\ell,c] \!=\!
	\left\{
	\begin{array}{ll}
			\!\!\!\frac{1}{N}\beta_i(c) e^{-j2\pi\frac{(\ell-\ell')(k'+\kappa')}{MN}}
			    								 & \hspace{-6em}    \ell'\leq \ell < M  \\
			\!\!\!\frac{1}{N}(\beta_i(c)\!-\!1) e^{-j2\pi\frac{(\ell\!-\ell')(k'+\kappa')}{MN}}e^{-j2\pi\frac{(k-k'+c)_N}{N}}
			\nonumber\\
										& \hspace{-6em}  0\leq \ell <\ell', \end{array} \right.
\end{align}
where $\beta_i(c) = \frac{ e^{-j2\pi(-c-\kappa')} -1}{e^{-j\frac{2\pi}{N}(-c-\kappa')} -1}$ and $\kappa'$ denotes the fractional Doppler associated with the $(k',\ell')$ path. Moreover, in~\eqref{eq:yAP:rec:frac}, $ w_p[k,\ell]$ is the received additive noise, which by using~\eqref{eq:rt} and~\eqref{eq:Ymn} can be expressed as
\vspace{-0.2em}
\begin{align*}
w_p[k,\ell] = \frac{1}{\sqrt{MN}}\sum_{n=0}^{N-1}\sum_{m=0}^{M-1} W_p[n,m] e^{-j2\pi(\frac{nk}{N}-\frac{m\ell}{M})},
\end{align*}
where $W_p[n,m]$ is the received noise sampled at $t=nT$ and $\nu=m\Delta f$, given by $W_p[n,m]=\int w_p(t) g_{rx}(t-nT)
e^{-j2\pi m \Delta f (t-nT)} dt.$
It can be readily checked that since $w_p(t)\sim\mathcal{CN}(0,\Sn)$, we have also  $w_p[k,\ell]\sim\mathcal{CN}(0,\Sn)$.

For the sake of simplicity of presentation and analysis, we consider the vector form representation of the input-output relationship of OTFS system in the DD domain. Let $\qv\in\{\qx_q, \qy_p, \qw_p\}\in \mathcal{C}^{MN\times 1}$ where $\qx_q$, $\qy_p$, and $\qw_p$ denote the vector of transmitted symbols from the $q$th user, the received signal vector at the $p$th AP, and the corresponding noise vector, respectively.
Hence, the input-output relationship in~\eqref{eq:yAP:rec:frac} can be expressed in vector form as
\vspace{-0.4em}
\begin{align}~\label{eq:yAPm:Vect}
\qy_p = \sum_{q=1}^{K_u} \sqrt{\rho_u \eta_q}\qH_{pq} \qx_q + \qw_p
\end{align}
where  $\qH_{pq}\in \mathcal{C}^{MN\times MN}$ is the effective DD domain channel between the $q$th user and $p$th AP, given by~\cite{KWAN:TWC:2021}
\vspace{-0.4em}
\begin{align}~\label{eq:Hpq}
\qH_{pq}
&=
\sum_{i=1}^{\Lmk}
\hmki \qT_{pq}^{(i)},
\end{align}
where $\qT_{pq}^{(i)}=(\textbf{F}_N \otimes \textbf{I}_M)
\boldsymbol{\Pi}^{\Lmki} \boldsymbol{\Delta }^{\Kmki+\kpmki}
(\textbf{F}_N^\dag \otimes \textbf{I}_M)$  and $\boldsymbol{\Pi}$ denotes a $MN\times MN$ permutation matrix, given in~\cite{KWAN:TWC:2021}, and $\boldsymbol{\Delta} = \diag\{z^0,z^1,\ldots,z^{MN-1}\}$ is a diagonal matrix with $z=e^{\frac{j2\pi}{MN}}$.

\textit{Channel estimation: }
By considering TDD operation, we rely on channel reciprocity to acquire channel state information (CSI). An intuitive method to estimate CSI is to transmit an impulse in the DD domain as the training pilot and then estimate the DD channel impulse response using the least square (LS) or MMSE estimator~\cite{Raviteja:TVT:2019}. Transmitted impulse pilots are spread by the channel and interfere with data symbols in the DD domain. Therefore, insertion of guard symbols to avoid the interference between the pilot and data symbols is required~\cite{Raviteja:TVT:2019,Flanagan:vtc:2020}. This technique, however, incurs a huge pilot overhead in cell-free massive MIMO. For instance, in the considered system, $K_u$ impulses are required to be transmitted. Assume that the DD channel responses of $K_u$ users have a finite support $[0, \ell_{max}]$ along the delay dimension and $[-k_{max}, k_{max}]$ along the Doppler dimension ($\ell_{max} = \max_{q}\ell_{pq}$ and $k_{max} = \max_{q}k_{pq}$). Then, the guard intervals between two adjacent impulse along the Doppler and delay dimension should not be smaller than $2k_{max}$ and $\ell_{max}$, respectively. Moreover, users cannot use dedicated pilot and guard grids of each other for data transmission.  As a result, the pilot length to transmit $K_u$ impulses in OTFS-based cell-free massive MIMO systems should be $2K_uk_{max}\ell_{max}$. This would be more challenging in the case of fractional Doppler, where by using full-guard pilot pattern~\cite{Raviteja:TVT:2019} the length of pilot overhead should be $2K_u N \ell_{max}$. Therefore, with a large number of users, the pilot overhead would be prohibitively high.

As an alternative solution, we deploy embedded-pilot channel estimation method with reduced guard symbols, while users are allowed to use dedicated pilot and guard DD grids of each other for data transmission. Consider $\varphi_q$ with $\mathbb{E}\{|\varphi_q|^2\}=P_p$ denoting a known pilot symbol for the $q$th user at a specific DD grid location $[k_q, \ell_q]$. Let  $x_{dq}[k,\ell]$ denote the data symbol at grid point $[k,\ell]$, and assume that each pilot is surrounded by a guard region of zero symbols.  Therefore, for the $q$th user the pilot, guard, and data symbols in the DD grid are arranged as
\begin{align}~\label{eq:datapatern}
x_q[k,\ell] = \left\{ \begin{array}{ll}
\!\!\!\varphi_q         &  k=k_q, \ell=\ell_q,  \\
\!\!\!0           &  k\in \mathcal{K}, k\neq k_q ~  \ell\in\mathcal{L}, \ell\neq\ell_q\\
\!\!\!x_{dq}[k,\ell] &  \mbox{otherwise}.\end{array} \right.
\end{align}
where $\mathcal{K}=\{k_q-2k_{max}-2\hat{k},\cdots, k_q+2k_{max}+2\hat{k}\}$, $\mathcal{L}=\{\ell_q-\ell_{max}\leq \ell \leq \ell_q+\ell_{max}\}$ and $\hat{k}$ denotes the additional guard to mitigate the spread due to fractional Doppler shift~\cite{Raviteja:TVT:2019}. In this case, the total overhead is $\Nguard=(2\ell_{max}+1)(4k_{max}+4\hat{k}+1)$ per each user. At the receiver,  the received symbols $y_p[k,\ell]$, $k_q-k_{max}-\hat{k} <k< k_q+k_{max}+\hat{k}$, $\ell_q\leq \ell \leq \ell_q-\ell_{max}$ are used for channel estimation. Therefore, from~\eqref{eq:yAP:rec:frac}, we have
\vspace{-0.3em}
\begin{align}~\label{eq:yAPpil}
y_p[k,\ell] &=\sqrt{P_p}
\tilde{b}[\ell-\ell_q] \tilde{h}_{pq}[(k-k_q)_N,(\ell-\ell_q)_M]\varphi_q\nonumber\\
&\hspace{2em}
+\mathcal{I}_1(k,\ell)+\mathcal{I}_2(k,\ell)+ w[k,\ell],
\end{align}
where
\vspace{-0.3em}
\begin{align*}
\tilde{b}[\ell-\ell_q] =
\left\{ \begin{array}{ll}
1,        &  \sum_{k'=0}^{k_{pq}} b[k',\ell-\ell_q]\geq 1  \\
0           &  \mbox{otherwise},\end{array} \right.
\end{align*}
is the path indicator and
\vspace{-0.4em}
\begin{align*}
\tilde{h}_{pq}[(k\!-\!k_q)_N,\!(\ell\!-\!\ell_q)_M\!] \!\!=\!\!
\!\sum_{k'=0}^{k_{pq}}\!\!
b[k'\!,\ell\!-\!\ell_q] h_{pq}[k'\!,\!\ell\!-\!\ell_q]\alpha(k,\ell,c).
\end{align*}

In~\eqref{eq:yAPpil}, $\mathcal{I}_1(k,\ell)$ denotes the interference spread from the $q$th user's data symbols due to fractional Doppler
\vspace{-0.3em}
\begin{align}~\label{eq:Ikl1}
\mathcal{I}_1(k,\ell) &\!=\!
\sqrt{\rho_u}\!
\sum_{k'=0}^{k_{pq}}
\sum_{\ell'=0}^{\ell_{pq}}
b[k',\ell']
\sum_{c\not\in \mathcal{K}}
h_{pq}[(k\!-\!k')_N,(\ell\!-\!\ell')_M]\nonumber\\
&\hspace{1em}\times
\alpha[k,\ell,c]x_{dq}[(k-k'+c)_N,(\ell-\ell')_M],
\end{align}
and  $\mathcal{I}_2(k,\ell)$ denotes the interference spread from the data symbols of other users, which can be expressed as
\vspace{-0.3em}
\begin{align}~\label{eq:Ikl2}
\mathcal{I}_2(k,\ell) &= \sqrt{\rho_u}\sum_{q'\neq q}^{K_u}
\sum_{k'=0}^{k_{pq'}}
\sum_{\ell'=0}^{\ell_{pq'}}
b[k',\ell']
\sum_{c=-N/2}^{N/2}
 h_{pq'}[k',\ell']\nonumber\\
 &\hspace{1em}\times
 \alpha(k,\ell,c)x_{dq'}[(k-k'+c)_N,(\ell-\ell')_M].
\end{align}

Assume that the number of paths, the index of delay and Doppler shifts are known by using the channel path estimation method proposed in~\cite{Flanagan:vtc:2020}. Then, by using the MMSE,  $h_{pq,i}$ can be estimated as $\hat{h}_{pq,i} =  c_{pq,i}y_p,$ where
\vspace{-0.2em}
\begin{align}~\label{eq:cpqk}
c_{pq,i}\!=\!
\frac{\sqrt{P_p}\beta_{pq,i}}
{P_p\beta_{pq,i}
	\!+\! \mathbb{E}
		\{|\mathcal{I}_1(k,\ell)|^2\} \!+\!
	  \mathbb{E}
		\{|\mathcal{I}_2(k,\ell)|^2\}\!+\!
	\Sn}.
\end{align}
By invoking~\eqref{eq:Ikl1}, and considering the fact that data symbols are independent,
$\mathbb{E}\{|x_{dq}[k',\ell']|^2\}=P_s$, and assuming independent channel gain,  we obtain
\vspace{-0.2em}
\begin{align}~\label{eq:EIk1}
\mathbb{E}\{|\mathcal{I}_1(k,\ell)|^2\}
&=
\rho_u
	\sum_{i=1}^{L_{pq}}
				\mathbb{E}
				\big\{
						\big|h_{pq,i}\big|^2
				\big\}				
					\sum_{c \notin \mathcal{K} }\big|\alpha[k,\ell,c]
				\big|^2.
\end{align}

We notice that in the delay domain, only $\ell_{pq}+1$ symbols before $\ell$ affect the received symbol in $\ell$. However, due to fractional Doppler shift, all data symbols outside the guard space $\mathcal{K}$ interfere with the received symbol on index $k$. It can be checked that for $k\in[k_q-k_{max}-\hat{k},k_q+k_{max}+\hat{k}]$ and $c \notin \mathcal{K}$,   $\big|\alpha[k,\ell,c]\big|$ becomes almost a constant. For the case of rectangular window, $\big|\alpha[k,\ell,c]\big|^2 \approx 1/N^2$~\cite{Kwan:TCOM:2021}. Therefore, we have $\sum_{c \notin \mathcal{K} }\big|\alpha[k,\ell,c]\big|^2 \approx \frac{(N-4k_{max}-4\hat{k}-1)}{N^2}$. Accordingly,~\eqref{eq:EIk1} can be expressed as
\vspace{-0.6em}
\begin{align}~\label{eq:EIk1:final}
\mathbb{E}\big\{\big|\mathcal{I}_1(k,\ell)\big|^2\big\} &\approx\frac{\rho_u(N-4k_{max}-4\hat{k}-1)}{N^2}
\sum_{i=1}^{\Lmk}\mathbb{E}\big\{\big|\hmki\big|^2\big\}.
\end{align}

By using similar steps, we can obtain
\vspace{-0.3em}
\begin{align}~\label{eq:EIk2:final}
\mathbb{E}\big\{\big|\mathcal{I}_2(k,\ell)\big|^2\big\} &\approx
\frac{\rho_u}{N}
\sum_{q'\neq q}^{K_u}
\sum_{i=1}^{\Lmk}\mathbb{E}\big\{\big|h_{pq',i}\big|^2\big\}.
\end{align}
To this end, by substituting~\eqref{eq:EIk1:final} and~\eqref{eq:EIk2:final} into~\eqref{eq:cpqk}, we get
\vspace{-0.1em}
\begin{align}~\label{eq:MMSE1:Final}
c_{pq,i} &= \frac{\sqrt{\rho_p}\beta_{pq,i}}{\rho_p\beta_{pq,i}\!
	\!+\! \rho_u \varXi+1},
\end{align}
where $\rho_p=P_p/\Sn$ is the normalized SNR of each pilot symbol and $\varXi=\frac{1}{N}\sum_{\substack{q'=1}}^{K_u}\sum_{i=1}^{L_{pq'}}\beta_{pq',i}
\! -\!\frac{(4k_{max}+4\hat{k}+1)}{N^2}\sum_{i=1}^{\Lmk}\beta_{pq,i}$.

\vspace{-0.2em}
\subsection{Downlink Payload Data Transmission}~\label{Sec:DLData}
The APs use conjugate beamforming to transmit signals to $K_u$ users. Let $s_{qr}=s_{q}[k,\ell]$, with $r=kM+\ell$, $ k\in \mathbb{N}[0,N-1], \ell\in \mathbb{N}[0,M-1]$ denote the i.i.d. DD domain information symbols to be transmitted to the $q$th user, which satisfies $s_{qr}\sim \mathcal{N}(0,1)$. Therefore, the signal transmitted   from the $p$th AP is
\vspace{-0.4em}
\begin{align}~\label{eq:xqd}
\qx_{d,p}= \sqrt{\rho_d}
\sum_{q=1}^{K_u}
\eta_{pq}^{1/2}\hat{\qH}_{pq}^\dag \qs_q,
\end{align}
where $\rho_d$ is the normalized SNR of each symbol; $\hat{\qH}_{pq}$ denotes the estimated channel matrix between the $q$th user and $p$th AP; $\qs_q\in\mathcal{C}^{MN\times 1}$ is the intended signal vector for the $q$th user; $\eta_{pq}$, $p=1,\ldots,M_a$, $q=1,\ldots,K_u$ are the power control coefficients chosen to satisfy  the  following power constraint at each AP~\cite{Hien:cellfree}
\vspace{-0.6em}
\begin{align}~\label{eq:AP:powcons}
\mathbb{E}\left\{\|\qx_{d,p}\|^2\right\} \leq \rho_d.
\end{align}

The received signal at the $q$th user in DD domain can be expressed as
\vspace{-1em}
\begin{align}~\label{eq:zqd}
\qz_{d,q}
&=\sqrt{\rho_d}\sum_{p=1}^{M_a}
\eta_{pq}^{1/2} {\qH}_{pq} \hat{\qH}_{pq}^\dag \qs_{q}\nonumber\\
&\hspace{2em}+
\sqrt{\rho_d}\sum_{p=1}^{M_a}
\sum_{q'\neq q}^{K_u}
\eta_{pq'}^{1/2}{\qH}_{pq}\hat{\qH}_{pq'}^\dag \qs_{q'} +\qw_{d,q},
\end{align}
where $\qw_{d,q}\in\mathcal{C}^{MN\times 1}$ is the AWGN vector at the user $q$.

\section{Performance Analysis}~\label{Sec:Perf}
In this section, we derive a new closed-form expression for the downlink achievable rate, using the bounding technique from~\cite{Hien:cellfree}. We further assume that the proposed low complexity DD domain detector in~\cite{Raviteja:TVT:2021} is utilized for detection of information symbols at each user.

We start by presenting some characteristics of the DD domain channel representation in~\eqref{eq:Hpq}, which will facilitate the ensuing achievable rate analysis.

\begin{lemma}\label{lemma:Tqi}
For the matrix $\qT_{pq}^{(i)}$, we have $\qT_{pq}^{(i)}\qT_{pq}^{(i)^{\dag}}=\qI_{MN}$.
\end{lemma}
\begin{proof}
	The proof is relegated to the journal version.
\end{proof}

\begin{lemma}\label{lemma:TqTqp}
For any two different paths with different delay indices, we have $\Big[\qT_{pq}^{(i)} \qT_{pq'}^{(j)^\dag}\Big]_{(r,r)} =
0$ for $(\ell_{pq,i}-\ell_{pq,j})_M\neq 0$.
\end{lemma}

\begin{proof}
	The proof is relegated to the journal version.
\end{proof}

\begin{lemma}\label{lemma:TqTqp:abs}
 For any two matrices of $\qT_{pq}^{(i)}$ and $\qT_{pq'}^{(j)^\dag}$ in~\eqref{eq:Hpq}, we have $\Big|\sum_{r'=1}^{MN} \big[\qT_{pq}^{(i)} \qT_{pq'}^{(j)^\dag}\big]_{(r,r')}\Big|^2=1$, $\forall q, q', i, j$.
\end{lemma}
\begin{proof}
	The proof is relegated to the journal version.
\end{proof}

\begin{lemma}\label{lemma:power}
The power constraint in~\eqref{eq:AP:powcons}, is expressed as
\vspace{-0.0em}
\begin{align}~\label{eq:dlpowercont}
\sum_{q=1}^{K_u}\sum_{i=1}^{\Lmk}
\eta_{pq}\gamma_{pq,i}\leq 1,
\end{align}
where $\gamma_{pq,i} \triangleq \mathbb{E}\big\{|\hat{h}_{pq,i}|^2\big\} =\sqrt{\rho_p} \beta_{pq,i}c_{pq,i}$.
\end{lemma}

\begin{proof}
By invoking~\eqref{eq:xqd}, we have
\begin{align}~\label{eq:AP:powcons:final0}
&\mathbb{E}\left\{\|\qx_{d,p}\|^2\right\}=
\mathbb{E}
\bigg\{
\Big\|
\sqrt{\rho_d}
\sum_{q=1}^{K_u}
\eta_{pq}^{1/2}\hat{\qH}_{pq}^\dag \qs_q
\Big\|^2
\bigg\}
\nonumber\\
&\hspace{1em}\!=\!
\rho_d\mathbb{E}
\bigg\{\! \Trace
\bigg(
\Big[\sum_{q=1}^{K_u}\eta_{pq}^{1/2}\hat{\qH}_{pq}^\dag \qs_q
\Big]\!\!
\Big[\sum_{q'=1}^{K_u}\eta_{pq'}^{1/2}\qs_{q'}^\dag\hat{\qH}_{pq'}
\Big]
\bigg)
\!\bigg\}
\nonumber\\
&\hspace{1em}
=
\rho_d\sum_{q=1}^{K_u}
\eta_{pq}\Trace\left(\mathbb{E}\big\{\hat{\qH}_{pq}^\dag \hat{\qH}_{pq}\big\}\right),
\end{align}
where the equality holds since $s_{q,r}$ are assumed to be i.i.d. RVs. By using~\eqref{eq:Hpq}, and noticing that channel gains of different path are zero-mean i.i.d RVs, i.e., $\mathbb{E}\big\{\hat{h}_{pq,i}^* \hat{h}_{pq,j}\big\}=0$, $\forall i\neq j$, and then applying Lemma~\ref{lemma:Tqi}  we get $\mathbb{E}\left\{\|\qx_{d,q}\|^2\right\}= \rho_d\sum_{q=1}^{K_u}\sum_{i=1}^{\Lmk} \eta_{pq}\gamma_{pq,i}$ which completes the proof.
\end{proof}

\vspace{1em}
We assume that each user has knowledge of the channel statistics in DD domain but not of the channel realizations~\cite{Hien:cellfree}. The signal received at the $q$th user~\eqref{eq:zqd} can be re-arranged to be suitable for detection of the $r$th entry of the received signal in DD domain with only statistical channel knowledge at users, given by~\eqref{eq:zqr} at the top of the next page.
\bigformulatop{22}{
\begin{align}~\label{eq:zqr}
z_{d,qr}
\!=&\!
\sqrt{\rho_d}
\underbrace{
	\mathbb{E}\Big\{
	\!\sum_{p=1}^{M_a}\!
	\eta_{pq}^{1/2}[{\qH}_{pq}]_{(r,:)}
	[\hat{\qH}_{pq}^\dag]_{(:,r)}\!\Big\}}_{\text{desired signal}\triangleq\mathbb{DS}_q}s_{qr}
   \!+\!
\sqrt{\rho_d}
\underbrace{
	\Bigg(\!\sum_{p=1}^{M_a}\!
	\eta_{pq}^{1/2}
	[{\qH}_{pq}]_{(r,:)}
	[\hat{\qH}_{pq}^\dag]_{(:,r)}\!-\!\mathbb{E}\Big\{\sum_{p=1}^{M_a}
	\eta_{pq}^{1/2}
	[{\qH}_{pq}]_{(r,:)\!}
	[\hat{\qH}_{pq}^\dag]_{(:,r)}\!\Big\}\!\Bigg)}_{\text{precoding gain uncertainty}\triangleq \mathbb{BU}_q}s_{qr}\nonumber\\
&\hspace{1em}
\! +\!
\underbrace{\sqrt{\rho_d}\sum_{p=1}^{M_a}
	\sum_{r'\neq r}^{MN}\eta_{pq}^{1/2}
	[{\qH}_{pq}]_{(r,:)} [\hat{\qH}_{pq}^\dag]_{(:,r')}s_{qr'}}_{\text{inter-symbol interference}\triangleq \mathbb{I}_{q,1} }
+
\underbrace{\sqrt{\rho_d}\sum_{p=1}^{M_a}\sum_{q'\neq q}^{K_u}
	\sum_{\substack{r'=1 }}^{MN}\eta_{pq'}^{1/2}
	[{\qH}_{pq}]_{(r,:)} [\hat{\qH}_{pq'}^\dag]_{(:,r')} s_{q'r'}}_{\text{inter-user interference}\triangleq\mathbb{I}_{q,2} } +
\underbrace{w_{d,qr}}_\text{AWGN}.
\end{align}
}

The sum of second, third, forth and last term in~\eqref{eq:zqr} are treated as ``effective noise". Since $s_{qr}$, $r=1,\ldots,MN$ are i.i.d. RVs, it can be readily shown that the effective noise and desired signal are uncorrelated~\cite{Hien:cellfree}. Hence, by using the fact that uncorrelated Gaussian noise represents the worst case, we obtain the following achievable rate of the $q$th user for OTFS-based cell-free operation.

\begin{theorem}~\label{Prop:DL:rate}
An achievable downlink rate of the transmission from the APs to the $q$th user for any finite $M_a$ and $K_u$, is given by~\eqref{eq:Rdq:final} at the top of the next page, where
\bigformulatop{23}{
\begin{align}~\label{eq:Rdq:final}
&R_{d,q}
=\frac{1}{MN}
\sum_{r=1}^{MN}
  \log_2
     \left(1\!+\! \frac
{\rho_{d}\left(\sum_{p=1}^{M_a}	\sum_{i=1}^{\Lmk} \eta_{pq}^{1/2}\gamma_{pq,i} \right)^2}
{  \rho_{d}
\Big(
\sum_{p=1}^{M_a}
\eta_{pq}
\sum_{i=1}^{\Lmk}
\beta_{pq,i}
\Big(\sum_{j=1}^{\Lmk}
\gamma_{pq,j}(\chi_{q,ij} +\kappa_{q,ij})
+\sum_{q'\neq q}^{K_u}
\sum_{j=1}^{\Lmkp}\!
\frac{\eta_{pq'}}{\eta_{pq}}\gamma_{pq',j}
\Big)
+
1}\right),
\end{align}
}
$\chi_{pq,ij}=\big|[\qT_{pq}^{(i)}\qT_{pq}^{(j)^\dag}]_{(r,r)}\big|^2$, while $\kappa_{pq,ij}=\big|\sum_{r'\neq r}^{MN}\big[\qT_{pq}^{(i)}\qT_{pq}^{(j)^\dag}\big]_{(r,r')}\big|^2$.
\end{theorem}

\bigformulatop{35}{
\begin{align}~\label{eq:Rdq:final:asymp}
R_{d,q} &=
  \log_2
     \left(1+\frac{\rho_{d}\left(\sum_{p=1}^{M_a}	\sum_{i=1}^{\Lmk} \eta_{pq}^{1/2}\gamma_{pq,i} \right)^2}
     {\rho_{d}
	\sum_{p=1}^{M_a}
	\eta_{pq}
	\sum_{i=1}^{\Lmk}
	\beta_{pq,i}
   \Big(
	\sum_{j=1}^{\Lmk}\!
	\gamma_{pq,j} + \sum_{q'\neq q}^{K_u}
	\sum_{j=1}^{\Lmkp}\!
	\frac{\eta_{pq'}}{\eta_{pq}}\gamma_{pq',j}\Big)
	+
	1}\right).
\end{align}}
\begin{proof}
 Since the channel model in~\eqref{eq:Hpq} consists of simply $MN$ parallel channels, an achievable rate per
channel input at the $q$th user is given by
\vspace{-0.0em}
\setcounter{equation}{24}
\begin{align}~\label{eq:Rdq}
R_{d,q} = \frac{\sum_{r=1}^{MN} I_{q,r} (\mathtt{SINR}_{d,r})}{MN},
\end{align}
with $I_{q,r} (\mathtt{SINR}_{d,r})= \log_2(1+\mathtt{SINR}_{d,r})$, and $\mathtt{SINR}_{d,r}$, i.e., the signal-to-interference-plus-noise ratio (SINR) at the $q$th user is\footnote{We note that the channel overhead is ignored as the uplink achievable rate is not considered in this paper. We emphasize that this assumption does not affect the insights obtained in our work.}
\vspace{-0.0em}
\begin{align}~\label{eq:SINRdq}
\mathtt{SINR}_{d,r}\!=\!
\frac{
	\big|\mathbb{DS}_q
	\big|^2
     }
 {
 	\mathbb{E}
	\big\{ \big|\mathbb{BU}_q\big|^2 \big\}
	\! \!+\!
	 \mathbb{E}
	 \big\{\big|\mathbb{I}_{q,1}\big|^2\big\}
	 \!+\!\!
	 \mathbb{E}
	 \big\{\big|\mathbb{I}_{q,2}\big|^2\big\}
	 \! +\!\!
1/\rho_d    }.
\end{align}

We now proceed to derive the achievable rate $R_{d,q}$. Noticing that $\qH_{pq}=\sum_{i=1}^{\Lmk}\!
\hmki \qT_{pq}^{(i)}$ and $\hat{\qH}_{pq}=\sum_{i=1}^{\Lmk}\!\hat{h}_{pq,i} \qT_{pq}^{(i)}$, the term in the numerator of~\eqref{eq:SINRdq} can be derived as
\vspace{2em}
\begin{align}~\label{eq:DS}
\mathbb{DS}_q
&\stackrel{(a)}{=}\!
      \sum_{p=1}^{M_a}\!
	            \eta_{pq}^{1/2}
				\mathbb{E}\Bigg\{\!
								\!\Big(
								\sum_{i=1}^{\Lmk}\!
								\hmkihat [\qT_{pq}^{(i)}]_{(r,:)}
								\!\Big) \!
								\Big(\!
								\sum_{j=1}^{\Lmk}\!
								\hmkjhatc [\qT_{pq}^{(j)^\dag}]_{(:,r)}
								\!\Big)
							\!\Bigg\}
		 \nonumber\\
&\hspace{0em}\stackrel{(b)}{=}
		\sum_{p=1}^{M_a}
		\eta_{pq}^{1/2}
		\bigg(
			\sum_{i=1}^{\Lmk}
			\mathbb{E}\left\{
							|\hmkihat|^2
							\Big| [\qT_{pq}^{(i)}]_{(r,r)}\Big|^2
						\right\}
		\bigg)
		\nonumber\\
&\hspace{0em}\stackrel{(c)}{=}
\sum_{p=1}^{M_a}	
		\sum_{i=1}^{\Lmk}
		\eta_{pq}^{1/2}\gamma_{pq,i},
\end{align}
where (a) follows by substituting $\varepsilon_{pq,i}=\hmki-\hmkihat$ and then using the fact that $\varepsilon_{pq,i}$ and $\hmkihat$ are independent RVs and $	\mathbb{E}\{\varepsilon_{pq,i}\}=0$; (b) follows from the fact that $\hmkihat$ are zero mean and independent; (c) follows from Lemma~\ref{lemma:Tqi}.

By using the fact that the variance of a sum of independent RVs is equal to the sum of the variances, $\mathbb{E}
\{ |\mathbb{BU}_q|^2 \}$ in~\eqref{eq:SINRdq} can be  derived as
\vspace{-0.0em}
\begin{align}~\label{eq:vardd}
\mathbb{E}
\Big\{ \big|\mathbb{BU}_q\big|^2 \Big\} &=
	\sum_{p=1}^{M_a}
				\eta_{pq}
						\Bigg(
							\mathbb{E}
								\Big\{	
									\Big|[{\qH}_{pq}]_{(r,:)}
                                              [\hat{\qH}_{pq}^\dag]_{(:,r)}
                                     \Big|^2
                                 \Big\}
                                  \nonumber\\
&\hspace{2em}
-\Big| \mathbb{E}\Big\{ [{\qH}_{pq}]_{(r,:)}
[\hat{\qH}_{pq}^\dag]_{(:,r)}\Big\}\Big|^2\Bigg).
\end{align}

Now, by using~\eqref{eq:Hpq}, and then applying Lemma~\ref{lemma:Tqi}, i.e., $\big[\qT_{pq}^{(i)}\qT_{pq}^{(i)^\dag}\big]_{(r,r)}=1$,~\eqref{eq:vardd} can be expressed as
\vspace{-0.0em}
\begin{align}~\label{eq:vard}
&\mathbb{E}
\Big\{ \big|\mathbb{BU}_q\big|^2 \Big\}
&=
\sum_{p=1}^{M_a}
\eta_{pq}
\bigg(\mathbb{V}_1-
\Big(\sum_{i=1}^{\Lmk}
\gamma_{pq,i}\Big)^2
\bigg),
\end{align}
where
\begin{align*}
\mathbb{V}_1=
\mathbb{E}
	\Bigg\{
			\bigg|\!
				\sum_{i=1}^{\Lmk}\hmki\hat{h}_{pq,i}^*\!+\!\!
				\sum_{i=1}^{\Lmk}\!\sum_{j\neq i}^{\Lmk}
				[\qT_{pq}^{(i)}\qT_{pq}^{(j)^\dag}]_{(r,r)}\hmki\hmkjhatc
			\bigg|^2
	\Bigg\}.
\end{align*}
Before proceeding to derive $\mathbb{V}_1$, we define $X\triangleq\sum_{i=1}^{\Lmk}\hmki\hat{h}_{pq,i}^*$ and $Y\triangleq\sum_{i=1}^{\Lmk}\sum_{j\neq i}^{\Lmk}\hmki\hmkjhatc$. We note that $\mathbb{E}\left\{ |X+Y|^2\right\} = \mathbb{E}\left\{ |X|^2\right\} + \mathbb{E}\left\{ |Y|^2\right\}$ if $X$ and $Y$ are independent RVs and $\mathbb{E}\left\{Y\right\}=0$. It can be readily checked that $Y$ is a zero mean RV which is independent from $X$. Hence,  $\mathbb{V}_1$ can be evaluated as
\begin{align}~\label{eq:EV10}
\mathbb{V}_1
&=  \sum_{i=1}^{\Lmk}
  				\left(
					\mathbb{E}
							\left\{
								\left|\varepsilon_{pq,i}\hat{h}_{pq,i}^*
								\right|^2
						\right\}\!
						+\!
					\mathbb{E}
						 \left\{
			 				\left|\hat{h}_{pq,i}\right|^4
			 			\right\}
 				 \right)
 				 \nonumber\\
&\hspace{-1em}
+ \sum_{i=1}^{\Lmk}	
			\sum_{j\neq i}^{\Lmk}
					\mathbb{E}
					\left\{
							\left|\hat{h}_{pq,i}\right|^2
					\right\}
					\mathbb{E}
					\left\{
							\left|\hat{h}_{pq,j}\right|^2
					\right\}
					 +
					\sum_{i=1}^{\Lmk}
					\sum_{j\neq i}^{\Lmk}
					\mathbb{E}
					\left\{
					 		\mathfrak{R}[\Psi]
					\right\}
\nonumber\\
&\hspace{-1em}
 +
 \sum_{i=1}^{\Lmk}\!
 \sum_{j\neq i}^{\Lmk}\!
				 \chi_{q.ij}
 						\mathbb{E}
 						\left\{\!
 								\left|\varepsilon_{pq,i}\right|^2
 								\!\!+\!\!
 								\left|\hat{h}_{pq,i}\right|^2
 						\!\right\}
 						\mathbb{E}
 						\left\{\!
 								\left|\hat{h}_{pq,j}\right|^2
 						\!\right\}, 						
\end{align}
where $\Psi\!=\!\bigg[\Big([\qT_{pq}^{(i)}\qT_q{pq}^{(j)^\dag}]_{(r,r)} [\qT_{pq}^{(j)}\qT_{pq}^{(i)^\dag}]_{(r,r)}^*\Big)
\big(\hat{h}_{pq,i}\hat{h}_{pq,j}^{*}\big)^2 \bigg]$. To this end, by using the facts that $ \mathbb{E}\big\{\big|\varepsilon_{pq,i}\big|^2\big\} = \beta_{pq,i}\!-\!\gamma_{pq,i} $ and $\mathbb{E}\big\{\mathfrak{Re}(\Psi)\big\} =0$,~\eqref{eq:EV10} is reduced to
\vspace{-0.0em}
\begin{align}\label{eq:EV1}
\mathbb{V}_1=&
\sum_{i=1}^{\Lmk}\!
\left(
\gamma_{pq,i} (\beta_{pq,i}\!-\!\gamma_{pq,i})
\!+\! 2\gamma_{pq,i}^2\right)
\nonumber\\
&\!+\! \sum_{i=1}^{\Lmk} \sum_{j\neq i}^{\Lmk}\!
\left(\gamma_{pq,i}
\gamma_{pq,j}
\!+\!
\chi_{q,ij}\beta_{pq,i}\gamma_{pq,j}\right).
\end{align}
 Therefore, by substituting~\eqref{eq:EV1} into~\eqref{eq:vardd} we obtain
\vspace{-0.0em}
\begin{align}~\label{eq:vard:final}
&\mathbb{E}
\Big\{ \big|\mathbb{BU}_q\big|^2 \Big\}\! =\!\!
\sum_{p=1}^{M_a}\eta_{pq}\!\sum_{i=1}^{\Lmk}
\beta_{pq,i}\Big(\gamma_{pq,i}
\!+\!
\sum_{j\neq i}^{\Lmk}\!
\chi_{q,ij}\gamma_{pq,j}
\!\Big).
\end{align}

The inter-symbol interference term can be obtained as
\vspace{-0.0em}
\begin{align}
\mathbb{E}
\big\{|\mathbb{I}_{q,1}|^2\big\}&=
	\mathbb{E}
			\bigg\{
					\bigg|
						\sum_{p=1}^{M_a}
							\sum_{r'\neq r}^{MN}\eta_{pq}^{1/2}
								[{\qH}_{pq}]_{(r,:)} [\hat{\qH}_{pq}^\dag]_{(:,r')}
					\bigg|^2
			\bigg\} \nonumber\\
&\hspace{-2.5em} =
	\mathbb{E}
			\Bigg\{\!
				  \Bigg|
						\sum_{p=1}^{M_a} \!
						\sum_{r'\neq r}^{MN}\! \!
						\eta_{pq}^{1/2}
						\Big(\!
						       \sum_{i=1}^{\Lmk}\!
								\hmki \hat{h}_{pq,i}^*[\qT_{pq}^{(i)}]_{(r,:)} [\qT_{pq}^{(i)^\dag}\!]_{(:,r')}
\nonumber\\
&\hspace{-1em}+ \!
				\sum_{i=1}^{\Lmk}\!
				\sum_{ j\neq i}^{\Lmk}\!
							\hmki \hmkjhatc
							[\qT_{pq}^{(i)}]_{(r,:)} [\qT_{pq}^{(j)^\dag}]_{(:,r')}
						\Big)
				\Bigg|^2
		  \Bigg\}.
\end{align}
To this end, by using Lemma~\ref{lemma:Tqi}, i.e.,  $[\qT_{pq}^{(i)}\qT_{pq}^{(i)^\dag}]_{(r,r')}=0$, and then by substituting $h_{pq,i}=\varepsilon_{pq,i}+ \hat{h}_{pq,i}$,  we get
\vspace{-0.3em}
\begin{align}~\label{eq:Iq1}
\mathbb{E}
\big\{|\mathbb{I}_{q,1}|^2\big\}
&=
\sum_{p=1}^{M_a}
	\eta_{pq}
			\sum_{i=1}^{\Lmk}\!
			\sum_{j\neq i}^{\Lmk}\!
			\kappa_{q,ij}
			\bigg(
				\mathbb{E}\Big\{\Big|\varepsilon_{pq,i} \Big|^2\Big\}
				\mathbb{E}\Big\{\Big| \hat{h}_{pq,j}\Big|^2\Big\}
\nonumber\\
&
\hspace{7em}
+
				\mathbb{E}\Big\{\Big|\hat{h}_{pq,i}\Big|^2\Big\}
				\mathbb{E}\Big\{\Big|\hat{h}_{pq,j}\Big|^2\Big\}
			\bigg)
\nonumber\\
&=
\sum_{p=1}^{M_a}
\sum_{i=1}^{\Lmk}\!
\sum_{j\neq i}^{\Lmk}\!
\eta_{pq}
\kappa_{q,ij}
\beta_{pq,i}
\gamma_{pq,j},
\end{align}
where the first equality holds as the variance of a sum of independent RVs is equal to the sum of the variances and the second equality holds since $\varepsilon_{pq,i}$ has zero mean and is independent of $\hat{h}_{pq,i}$.

Noticing that the channel gains of different users are independent zero mean RVs, and applying Lemma.~\ref{lemma:TqTqp:abs}, the inter-user interference term can be derived as
\vspace{-0.3em}
\begin{align}~\label{eq:Iq2}
\mathbb{E}
\big\{|\mathbb{I}_{q,2}|^2\big\}
&=
\sum_{p=1}^{M_a}
\sum_{q'\neq q}^{K_u}
\sum_{i=1}^{\Lmk}\!
\sum_{j=1}^{\Lmkp}\!
					\eta_{pq'}\beta_{pq,i}\gamma_{pq',j}.
\end{align}
To this end, by substituting~\eqref{eq:DS},~\eqref{eq:vard:final},~\eqref{eq:Iq1}, and~\eqref{eq:Iq2} into~\eqref{eq:SINRdq}, after some algebraic manipulations the desired result in~\eqref{eq:Rdq:final} is obtained.
\end{proof}

\begin{corollary}~\label{corr:dl:apx}
In the special case that different paths corresponding to the channel between the $q$th user and $p$th AP experience different delays (i.e., $\Lmki\neq\ell_{pq,j}$), the achievable downlink rate in~\eqref{eq:Rdq} reduces to~\eqref{eq:Rdq:final:asymp} at the top of the page.
\end{corollary}
\begin{proof}
By invoking Lemma~\ref{lemma:TqTqp} and Lemma~\ref{lemma:TqTqp:abs}, it can be checked that when $\Lmki\neq\ell_{pq,j}$, we have $\chi_{q,ij}=0$ and $\kappa_{p,ij}=1$. Thus, the received SINR becomes independent of $r$, which completes the proof.
\end{proof}

\begin{remark}
By inspecting~\eqref{eq:Rdq:final:asymp}, we see that with perfect knowledge of delay and Doppler indices and using conjugate beamforming, the impact of delay shift and Doppler spread can be efficiently mitigated in the special case of $\Lmki\neq\ell_{pq,j}$.
\end{remark}

\vspace{-0.64em}
\section{Numerical Results and Discussions}~\label{Sec:Numer}
We consider an OTFS system with $N=20$ and $M=30$. The carrier frequency is $f_c=4$ GHz and the sub-carrier spacing is $\Delta f=15$ kHz. The maximum moving speed in the scenario is $500$ kmph, thus, the maximum Doppler index is $k_{max}=3$. We set maximum delay index as $\ell_{max}=2$ and the  number of channel paths as $L_{pq}=5$. The associated delay and Doppler indices for each channel path are randomly chosen from $[0,\ell_{max}]$ and $[-k_{max},k_{max}]$, respectively.

We assume that $M_a$ APs and $K_u$ users are uniformly distributed at random within a square of size $D \times D~\text{km}^2$ whose edges
are wrapped around to avoid the boundary effects. The large-scale fading coefficient $\beta_{pq}$ models the path loss
and shadow fading, according to $\beta_{pq,i} = \mathrm{PL}_{pq,i} 10 ^{\frac{\sigma_{sh} z_{pq,i}}{10}}$, where $\mathrm{PL}_{pq,i}$ represents the path loss, and $10 ^{\frac{\sigma_{sh} z_{pq,i}}{10}}$ represents the shadow fading with the standard deviation $\sigma_{sh}$, and $z_{pq,i} \sim \mathcal{CN}(0,1)$. We use the three-slope model for the path-loss $\mathrm{PL}_{pq,i}$ (in dB) as
\begin{align*}
\mathrm{PL}_{pq,i} \!=\!\! \left\{ \begin{array}{ll}
\!\!\!\!\!-L\!-35\log_{10}(d_{pq})         & d_{pq}>d_1,  \\
\!\!\!\!\!-L\!-15\log_{10}(d_{1})\!-20\log_{10}(d_{pq})            &  d_0<d_{pq}\leq d_1\\
\!\!\!\!\!-L\!-15\log_{10}(d_{1})\!-20\log_{10}(d_{0}) &  d_{pq}\leq d_0.\end{array} \right.
\end{align*}
where $L$ is a constant depending on the carrier frequency, the
user and AP heights, given in~\cite{Hien:cellfree}. We further use the correlated shadowing model for $d_{pq}>d_1$ as described in~\cite{Hien:cellfree}. Here, we choose $\sigma_{sh}=8$ dB, $D=1$ km, $d_1=50$ m, and $d_0=10$ m. We further set the noise figure $F = 9$ dB, and thus the noise power $\Sn=-108$ dBm ($\Sn=k_B T_0 (M\Delta f) F$ W, where $k_B=1.381 \times 10^{-23}$ Joules$/^{o}$K is the Boltzmann constant, while $T_0=290^o$K is the noise temperature). Let $\tilde{\rho}_d =1$ W, $\tilde{\rho}_u =0.2$ W and $\tilde{\rho}_p =1$ W be the maximum transmit power of the APs, users and uplink training pilots, respectively. The normalized maximum transmit power ${\rho}_d$, ${\rho}_u$, and ${\rho}_p$ are calculated by dividing these powers by the noise power $\Sn$. We assume that all APs transmit with full power, and at
the $p$th AP, the power control coefficients $\eta_{pq}$, $q=1,\ldots,K_u$, are the same, i.e.,
$\eta_{pq} = \big(\sum_{q'=1}^{K_u}\sum_{i=1}^{L_{pq}} \gamma_{pq',i}\big)^{-1}$, $\forall q=1,\ldots,K_u$.

\begin{figure}[t]
	\centering
	\vspace{0em}
	\includegraphics[width=97mm, height=70mm]{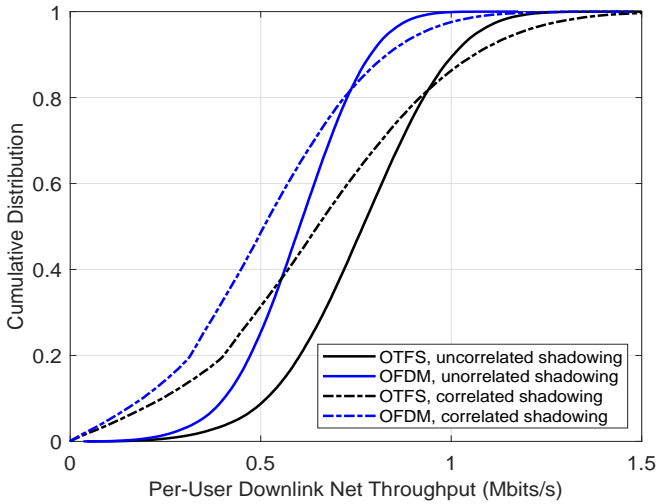}
	\vspace{-1em}
	\caption{The cumulative distribution of the per-user downlink throughput for correlated and uncorrelated shadow fading.}
	\vspace{-0.7em}
	\label{fig:Fig1}
\end{figure}

In Fig.~\ref{fig:Fig1}, we compare the performance of OTFS  with OFDM modulation system for $M_a=40$ and $K_u=20$ with and without shadow fading correlation. It is clear that  OTFS significantly outperforms OFDM in both median and in $95\%$-likely performance. The performance improvement of the OTFS over the OFDM is more pronounced in the uncorrelated shadowing case. More specifically, the $95\%$-likely throughput gain of the OTFS system over the OFDM counterpart is nearly $4$-fold greater in uncorrelated shadowing as compared to correlated shadowing.

Fig.~\ref{fig:Fig2} shows the average downlink throughput of the system versus number of APs and for two different number of users, with and without shadow fading correlation. We evaluate the SE of the system over 200 random realizations of the locations of APs, users and fading channels. The analytical results are based on Theorem~\ref{Prop:DL:rate}. It is observed that by increasing $K_u$, the per user rate decreases. This is because the same amount of resources are shared between more users. Moreover, the interference from other users increases, which degrades both the MMSE estimate of the channel gains and the received per-user SINR. This motivates  AP selection and user scheduling algorithm as well as pilot assignment algorithm design to manage the radio resources more efficiently.

\vspace{-0.0em}
\section{Conclusion}~\label{Sec:conclusion}
We analyzed the downlink performance of cell-free massive MIMO systems with OTFS modulation, taking into account the effects of channel estimation. A new closed-form expression for the average downlink throughput was presented. Our results confirmed the superiority of OTFS against OFDM in improving the throughput of the cell-free massive MIMO systems in high-mobility scenarios. Moreover,  the improvement in $95\%$-likely per-user throughput in uncorrelated shadowing is almost $4$ times higher than in the correlated fading scenarios.
\begin{figure}[t]
\centering
\vspace{0em}
\includegraphics[width=97mm, height=70mm]{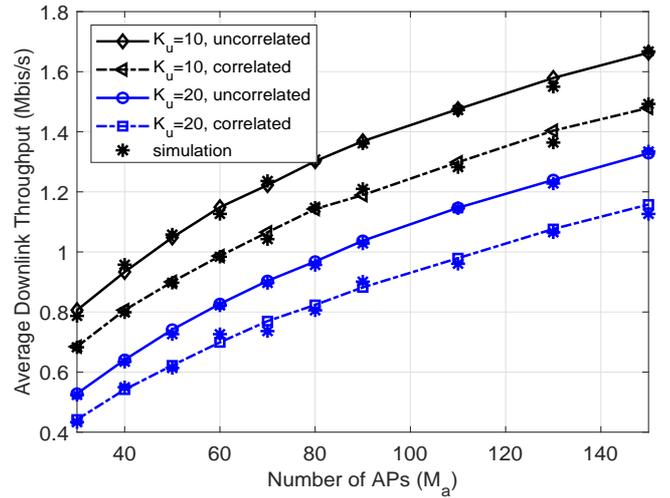}
\vspace{-1em}
\caption{The average downlink throughput versus the number of APs. }
\vspace{-0.7em}
\label{fig:Fig2}
\end{figure}
\vspace{0em}
\bibliographystyle{IEEEtran}



\begin{thebibliography}{1}

\bibitem{Jakes}
W. C. Jakes, \emph{Microwave Mobile Communications}, New York: IEEE Press, reprinted, 1994.


\bibitem{Hadani:WCNC:2017}
R. Hadani \emph{et al.},  ``Orthogonal time frequency space
modulation," in \emph{Proc. IEEE WCNC}, Mar. 2017.


\bibitem{Wei:WC:2021}
Z. Wei \emph{et al.}, ``Orthogonal time-frequency space modulation: A promising
next generation waveform," to appear in \emph{IEEE Wireless Commun.}, 2021.


\bibitem{Raviteja:TWC:2018}
P. Raviteja, K. T. Phan, Y. Hong, and E. Viterbo, ``Interference cancellation and iterative detection for orthogonal time frequency space modulation," \emph{IEEE Trans. Wireless Commun.,} vol. 17, no. 10, pp. 6501-6515, Oct. 2018.

\bibitem{Raviteja:TVT:2019}
P. Raviteja, K. T. Phan, and Y. Hong, ``Embedded pilot-aided channel estimation for OTFS in delay-Doppler channels," \emph{IEEE Trans. Veh. Technol.,} vol. 68, no. 5, pp. 4906-4917, May 2019.

\bibitem{Wang:JSAC:2020}
Y. Liu, S. Zhang, F. Gao, J. Ma, and X. Wang, ``Uplink-aided high mobility downlink channel estimation over massive MIMO-OTFS system,"
\emph{IEEE J. Sel. Areas Commun.,} vol. 38, no. 9, pp. 1994-2009, Sep. 2020.

\bibitem{Dobre:JSAC:2021}
M. Li, S. Zhang, F. Gao, P. Fan, and O. A. Dobre, ``A new path division multiple access for the massive MIMO-OTFS networks," \emph{IEEE J. Sel.
Areas Commun.,} vol. 39, no. 4, pp. 903-918, Apr. 2021.

\bibitem{Feng:ICC:2021}
J. Feng, H. Q. Ngo, M. F. Flanagan, and M. Matthaiou, ``Performance analysis of OTFS-based uplink massive MIMO with ZF receivers," in \emph{Proc. IEEE ICC}, May 2021.

\bibitem{Shi:TWC:2021}
D. Shi \emph{et al.}, ``Deterministic pilot design and channel estimation for downlink massive MIMO-OTFS systems in presence of the fractional Doppler," to appear in \emph{IEEE Trans. Wireless Commun.,} 2021.


\bibitem{Hien:cellfree}
H. Q. Ngo, A. Ashikhmin, H. Yang, E. G. Larsson, and T. L. Marzetta, ``Cell-free massive MIMO versus small cells," \emph{IEEE Trans. Wireless Commun.,} vol. 16, no. 3, pp. 1834-1850, Mar. 2017.

\bibitem{Nguyen:cellfree}
S.-N. Jin, D.-W. Yue, and H. H. Nguyen, ``Spectral efficiency of a frequency-selective cell-free massive MIMO system with phase noise,"
\emph{IEEE Wireless Commun. Lett.}, vol. 10, no. 3, pp. 483-487, Mar. 2021.

\bibitem{Schotten:cellfree}
W. Jiang and H. D. Schotten, ``Cell-free massive MIMO-OFDM transmission over frequency-selective fading channels," \emph{IEEE  Commun. Lett.}, vol. 25, no. 8, pp. 2718-2722, Aug. 2021.

\bibitem{KWAN:TWC:2021}
S. Li, J. Yuan, W. Yuan, Z. Wei, B. Bai, and D. W. K. Ng, ``Performance analysis of coded OTFS systems over high-mobility channels," \emph{IEEE Trans. Wireless Commun.,} vol. 20, no. 9, pp. 6033-6048, Sept. 2021.

\bibitem{Flanagan:vtc:2020}
V. Kumar Singh, M. F. Flanagan, and B. Cardiff, ``Maximum likelihood channel path detection and MMSE channel estimation in OTFS systems," in \emph{Proc. IEEE VTC,} Dec. 2020.

\bibitem{Raviteja:TVT:2021}
B. C. Pandey, S. K. Mohammed, P. Raviteja, Y. Hong, and E. Viterbo, ``Low complexity precoding and detection in multi-user massive MIMO OTFS downlink," \emph{IEEE Trans. Veh. Technol.,} vol. 70, no. 5, pp. 4389-4405, May 2021.

\bibitem{Kwan:TCOM:2021}
Z. Wei, W. Yuan, S. Li, J. Yuan, and D. W. K. Ng, ``Transmitter and receiver window designs for orthogonal time-frequency space modulation,"\textit{ IEEE Trans. Commun.,} vol. 69, no. 4, pp. 2207-2223, Apr. 2021.

\end{thebibliography}
\end{document}